\documentclass[aps,pre, reprint,superscriptaddress,nofootinbib]{revtex4-1}

\usepackage{bm, amsfonts, amsmath, amssymb, xspace, mathtools, amsthm}
\usepackage{graphicx}
\usepackage{hyperref}
\usepackage[dvipsnames]{xcolor}

\newcommand\myshade{100}
\colorlet{mylinkcolor}{BrickRed}
\colorlet{mycitecolor}{NavyBlue}
\colorlet{myurlcolor}{Aquamarine}

\hypersetup{
  linkcolor  = mylinkcolor!\myshade!black,
  citecolor  = mycitecolor!\myshade!black,
  urlcolor   = myurlcolor!\myshade!black,
  colorlinks = true,
}

\newtheorem{lemma}{\bf Lemma}[section]

\newcommand\PI[1]{\ensuremath{I_{\partial}^{#1}}\xspace}
\newcommand\IR[1]{\ensuremath{I_{\cap}^{#1}}\xspace}
\newcommand\bX{\ensuremath{\bm X_t}\xspace}
\xspace
\xspace
\xspace
\xspace

\newcommand{\meet}{\mathrel{\text{\raisebox{0.25ex}{\scalebox{0.8}{$\wedge$}}}}}

\newcommand{\splitatcommas}[1]{%
\begingroup
\begingroup\lccode`~=`, \lowercase{\endgroup
    \edef~{\mathchar\the\mathcode`, \penalty0 \noexpand\hspace{0pt plus 1em}}%
    }\mathcode`,="8000 #1%
    \endgroup
}

\begin{document}

\title{Improved estimators of causal emergence for large systems}%

\author{Madalina~I. Sas}
\affiliation{Centre for Complexity Science, Imperial College London}

\author{Fernando~E. Rosas}
\affiliation{Centre for Complexity Science, Imperial College London}
\affiliation{Centre for Psychedelic Research, Department of Brain Sciences, Imperial College London}
\affiliation{Sussex AI and Centre for Consciousness Science, Department of Informatics, University of Sussex}
\affiliation{Centre for Eudaimonia and Human Flourishing, University of Oxford}

\author{Hardik Rajpal}
\affiliation{Centre for Complexity Science, Imperial College London}
\affiliation{I-X Centre for AI in Science, Imperial College London}
\affiliation{Department of Electrical and Electronic Engineering, Imperial College London, London}

\author{Daniel Bor}
\affiliation{Department of Psychology, University of Cambridge}
\affiliation{Department of Psychology, Queen Mary University of London}

\author{Henrik~J. Jensen}
\affiliation{Centre for Complexity Science, Imperial College London}
\affiliation{Department of Mathematics, Imperial College London}

\author{Pedro~A.M. Mediano}
\affiliation{Department of Computing, Imperial College London}
\affiliation{Division of Psychology and Language Sciences, University College London}

\begin{abstract}

\noindent
A central challenge in the study of complex systems is the quantification of emergence --- understood as the ability of the system to exhibit collective behaviours that cannot be traced down to the individual components. 
While recent work has proposed practical measures to detect emergence, 
these approaches tend to double-count the contribution 
of shared components, which substantially hinders their capability to effectively study large systems. %
In this work, we introduce a family of improved information-theoretic measures of emergence that %
iteratively correct for double-counted terms. %
Our approach is computationally efficient and provides a controllable trade-off between computational load and sensitivity, leading to more accurate and versatile estimates of emergence.
The benefits of the proposed approach are demonstrated by successfully detecting emergence in both simulated and real-world data related to flocking behaviour.

\end{abstract}

\maketitle

\section{Introduction}

\noindent Emergence is one of the hallmarks of complex systems, %
where %
numerous interacting parts give rise to collective behaviour that could not be predicted by examining the parts in isolation~\cite{jensen2022complexity}.
Frameworks and approaches have been proposed to investigate emergence in various disciplines, including phase transitions in statistical physics~\cite{butterfield2012emergence}, 
cellular automata as a computational theory~\cite{mitchell1996evolving}, 
and the Schelling model in social sciences~\cite{schelling1971dynamic}. 
The study of emergence is also central to fundamental questions in biology, such as the origins of life~\cite{cohen2007explaining}, evolution~\cite{corning2012re}, and the neural basis of consciousness~\cite{broad1925}. 

Recent advances have enabled information-theoretic approaches to measuring emergence, allowing for the identification of emergent features in stochastic dynamical systems
~\cite{seth2010measuring,klein2020emergence,rosas2020reconciling,barnett2021dynamical,varley2022emergence,rosas2024software,yuan2024emergence}. 
These methods leverage information theory to characterise temporal 
interdependencies in the system across scales.
In particular, the \textit{mereological theory of causal emergence} formalises emergence through part-whole interactions~\cite{rosas2020reconciling,mediano2022greater}, interrogated through tools from the partial information decomposition (PID) framework~\cite{williams2010nonnegative,mediano2021towards}. 
This approach is particularly well-suited to quantify to what degree the `whole is greater than the sum of the parts' -- an identifying characteristic of emergent phenomena~\cite{anderson1972more,jensen2022complexity}.

Empirical applications of this framework to study emergence are based on a computationally tractable estimate known as $\Psi$. %
This metric assesses the difference between the information provided by a \textit{macroscopic} (systemic) property about its own future, and the sum of the information about it provided by each of its parts. 
This kind of ``whole-minus-sum'' measure is structurally similar to the synergy-redundancy index~\cite{gat1999synergy} and the O-information~\cite{rosas2019quantifying}, which characterise multivariate interdependencies through differences of mutual information terms -- a family of measures recently described as \textit{Shannon invariants}~\cite{gutknecht2025shannon}. 
This metric has been useful for a range of empirical investigations, including the study of gene regulatory networks~\cite{pigozzi2025associative}, the dynamics of the human brain~\cite{luppi2021reduced}, the internal dynamics of reservoir computing~\cite{tolle2024evolving}, and the formation of useful internal representations in machine learning~\cite{mcsharry2024learning}.

Despite their strengths, these measures tend to overestimate the contribution of individual components by counting shared information among them multiple times -- a feature that substantially hinders the sensitivity of the metric as the number of components grows. 
As a result, %
these measures tend to yield negative findings in large systems even when emergent phenomena are present. This is especially relevant for biological complex systems, where redundancy is to be expected alongside synergy for its functional role promoting robustness against uncertainty~\cite{krakauer2002redundancy,durant2019role,luppi2024information}. Hence, overestimating shared information can be particularly detrimental for investigating emergent phenomena in such systems.

To address this important limitation, %
we propose a new set of measures of causal emergence that retain sensitivity when applied to large systems. 
Continuing a tradition started with Shannon's work on information lattices~\cite{shannon_1953}, %
these new estimators are based on a novel method named \textit{lattice expansion}, which let us bridge between whole-minus-sum measures~\cite{rosas2019quantifying,rosas2025characterising}, which are fast but coarse, and the full partial information decomposition, which is precise but computationally demanding. 
The resulting measures improve the accuracy of existing estimators of emergence by iteratively correcting for the double-counting of shared information, 
thus allowing a more robust detection of emergence. %
The benefits of this approach are illustrated in case studies involving synthetic and real-world data. 
Overall, the generality of the method for correcting double-counted terms makes it suitable for application to other Shannon-invariant measures, opening the door to a wide range of improved estimators of multivariate interdependence.

\section{Technical background}

\noindent We consider a complex system composed by $n$ sub-components, which are measured regularly over timepoints $t\in \mathbb{N}$. The results of those measurements are the observables $\bm X_t =
(X_t^1,\dots,X_t^n)$, with $X_t^i\in \mathcal{X}_i$ corresponding to the state
of the $i$\textsuperscript{th} part at time $t$ in phase space
$\mathcal{X}_i$. We use $[n] \coloneqq \{1,\dots,n\}$ to refer to the indices of all parts.

Components of the system can be grouped in subsets $\alpha \in \mathcal A$, so that $\bm X_t^{a} = (X_t^{i_1},..., X_t^{i_K})$ represents the observables at time $t$ for a subset of components with indices in $a = \{i_1,\dots,i_K\}\subset [n]$. 
$\mathcal A$ denotes the powerset, or all the possible sets of combinations of parts.
Finally, let $\mathcal S^{(k)}$ represent the set of sets with more than $k$ parts: $\mathcal S^{(k)} = \{ \{\alpha_1, \dots, \alpha_L\} \subset \mathcal A : \min_j |a_j| > k \}$.

We consider two time points $t < t'$
and systemic macroscopic features $V_t\in \mathcal{V}$ that are \emph{supervenient} on the underlying system, and we formalise this condition by requiring $V_t$ to be statistically independent
of $\bm X_{t'}$ given $\bm X_t$ 
for all $t'> t$. This includes as
particular cases deterministic functions $V_t = F(\bm X_t)$, as well as aggregate properties %
such as coarse-grainings.
The information-theoretic emergence framework introduced by \cite{rosas2020reconciling} proposes measures and criteria to test whether $V$ is emergent and in what way using PID, which is briefly presented next. 

\subsection{Information decomposition}
\label{sec:pid}

\noindent At the core of information theory, the mutual information (MI) introduced by Shannon~\cite{shannon_1953} captures the extent to which knowing about one set of variables reduces uncertainty about another set. 
However, complex systems are often characterised by interactions which extend beyond pairs, and MI alone cannot attribute which variables within the set are responsible for the reduction in uncertainty.
To address this, the Partial Information Decomposition (PID) framework~\cite{williams2010nonnegative} proposes a decomposition of MI into different kinds of \textit{information atoms}, according to how multiple \textit{source} variables $X_i$ provide information about a \textit{target} $Y$. In the special case of two source variables, PID decomposes MI into four terms: 
\begin{align}\label{eqn:pid}
I(X_1, X_2; Y) = &~  \texttt{Red}(X_1, X_2; Y ) + \texttt{Un}(X_1; Y |X_2) \notag \\
& + \texttt{Un}(X_2; Y |X_1) + \texttt{Syn}(X_1, X_2; Y ) ~ ,
\end{align}
where
$\texttt{Un}(X_i;V)$ denotes \textit{unique} information provided by $X_i$,
$\texttt{Red}(X_1,X_2;V)$ is the \textit{redundant} information shared between sources, and 
$\texttt{Syn}(X_1,X_2;V)$ is the \textit{synergistic} information, present when sources are taken together, but not separately.

The decomposition of mutual information for more than two sources creates a hierarchy of information which can be expressed with the formalism of a \textit{redundancy lattice}, which captures a \textit{partial ordering} between information atoms, where a higher element provides at least as much redundant information as a lower one.
Generally, an atom or set of atoms $\alpha_1$ is less informative than another $\alpha_2$, if the information it contains is included in $\alpha_2$. We denote this as $\alpha_1\prec\alpha_2$. Conversely, two atoms $\alpha_1$ and $\beta$ may both be less informative than $\alpha_2$, but may not be comparable to each other (hence, the \textit{partial} ordering).

This allows us to construct a redundancy function $\IR{\bm\alpha}$, as the sum of all the preceding (lower) partial information atoms $\bm\beta$ in the lattice:
\begin{align}
  \IR{\bm\alpha}(\bX; Y) = \sum_{\bm\beta\preceq\bm\alpha} \PI{\bm\beta}(\bX; Y) 
  \label{eq:redpartial}
\end{align}

For instance, for three sources (see Fig.~\ref{fig:lattice}), the information provided in common by all three (represented by the atom $\{\{1\}\{2\}\{3\}\}$) is at the bottom of the lattice, as it is included in the information common to pairs of variables (e.g.~$\{\{1\}\{2\}\}$), which is included in the information within each variable. %
However, the number of atoms grows exponentially with the number of sources, making the estimation of information atoms computationally difficult.

Practically estimating information atoms relies on selecting a function for either synergy or redundancy. Most of the time, in practice, a redundancy function is used to compute the synergy. A broad range of such functions have been proposed in the literature~\cite{barrett2015exploration, harder2013bivariate, bertschinger2014quantifying, ince2017measuring, griffith2014quantifying, james2018unique}, see~\cite{lizier2018information} for a review on the topic. 

\subsection{Emergence from part-whole relationships}
\label{sec:theory}

\noindent The PID formalism is particularly well-suited to the use case of measures of emergence, as the sources can represent the microscopic, part-level variables, while the target can correspond to a macroscopic, system-level feature. The decomposition above allows the estimation of multivariate mutual information terms, such as $I(\bm X; V)$ for a system with many variables $X_i$, which would quantify the total information that the parts contribute towards a system property $V$, and could be a marker of $V$ being emergent if any synergistic information is present.

Moreover, by applying mutual information with a fixed time delay between two time series, the time-delayed mutual information (TDMI) can be used as a measure of predictability, allowing us to quantify how the past states -- micro or macro -- of a system may predict its future states.
When applied to the TDMI, the PID atoms acquire valuable interpretations which can contribute to understanding how information is stored, copied and transferred within the system over time~\cite{lizier2012local,mediano2019beyond}. 
These interpretations form the basis of the %
theory of causal emergence in~\cite{rosas2020reconciling} underlying the measures discussed in this paper.
We briefly recapitulate relevant measures and criteria below.

\textbf{Causal emergence}, quantified by $\Psi$, refers to the property of a system feature to be irreducible to the sum contributions of its parts: %
\begin{align}
    \Psi^{(k)}_{t,t'}(V)   & := I(V_t; V_{t'}) - \sum_{|\alpha|=k} I(\textbf X^{\alpha}_t; V_{t'}). \label{eq:psi}
\end{align}
Outside the context of this theory, $\Psi^{(1)}$ corresponds to
the well-known redundancy-synergy index~\cite{gat1999synergy}.%

\textbf{Downward causation}, quantified by $\Delta$, refers to a system where a macro feature has a causal effect over $k$ particular agents, but this effect cannot be attributed to any other individual component or group:
\begin{align}\label{eq:delta}
    \Delta^{(k)}_{t,t'}(V) & :=  \max_{|\alpha|=k} \left( I(V_t; \textbf X^{\alpha}_{t'}) - \sum_{|\beta|=k} I (\textbf X_t^{\beta}; \textbf X_{t'}^{\alpha}) \right).
\end{align}

\textbf{Causal decoupling}, quantified by $\Gamma$, refers to a system where a macro feature can
predict its own evolution, but no component or group of components may predict the
evolution of any other:

\begin{align}
    \Gamma^{(k)}_{t,t'} = \max_{|\alpha|=k} I(V_t, \textbf X_{t'}^\alpha).
\end{align}

The associated criteria for emergence are given by:

\begin{enumerate}
    \item[(1)]  $\Psi_{t,t'}^{(k)} (V) > 0 $ is sufficient for causal emergence.
    \item[(2)]  $\Delta^{(k)}_{t,t'} (V) > 0 $  is sufficient for downward causation.
    \item[(2)]  $\Gamma_{t,t'} = 0$ and  $\Psi_{t,t'}^{(k)} (V) > 0 $ is sufficient for causal decoupling.
\end{enumerate}

\section{Improved estimators of causal emergence}

\subsection{Motivation}

\noindent As outlined in Sec.~\ref{sec:theory}, the theory of causal emergence in~\cite{rosas2020reconciling}
provides practical criteria (e.g.\ $\Psi$) to detect causal emergence in
experimental or simulated data. Unfortunately, like other Shannon invariants, 
$\Psi$ is limited as it overestimates the contribution of the parts.
The negative sum of the marginal mutual informations in
Eq.~\eqref{eq:psi} contains multiple terms of redundant information, which are
``double-counted'' and subtracted in excess from $I(V_t; V_{t'})$. Thus, it
results in an estimator of ``synergy minus redundancy'', which may fail to
detect emergence if the double-counted redundancy exceeds
synergy.

\begin{figure*}[ht]
  \centering
  %\includetikzfig{LatticeExpansion}
  \includegraphics{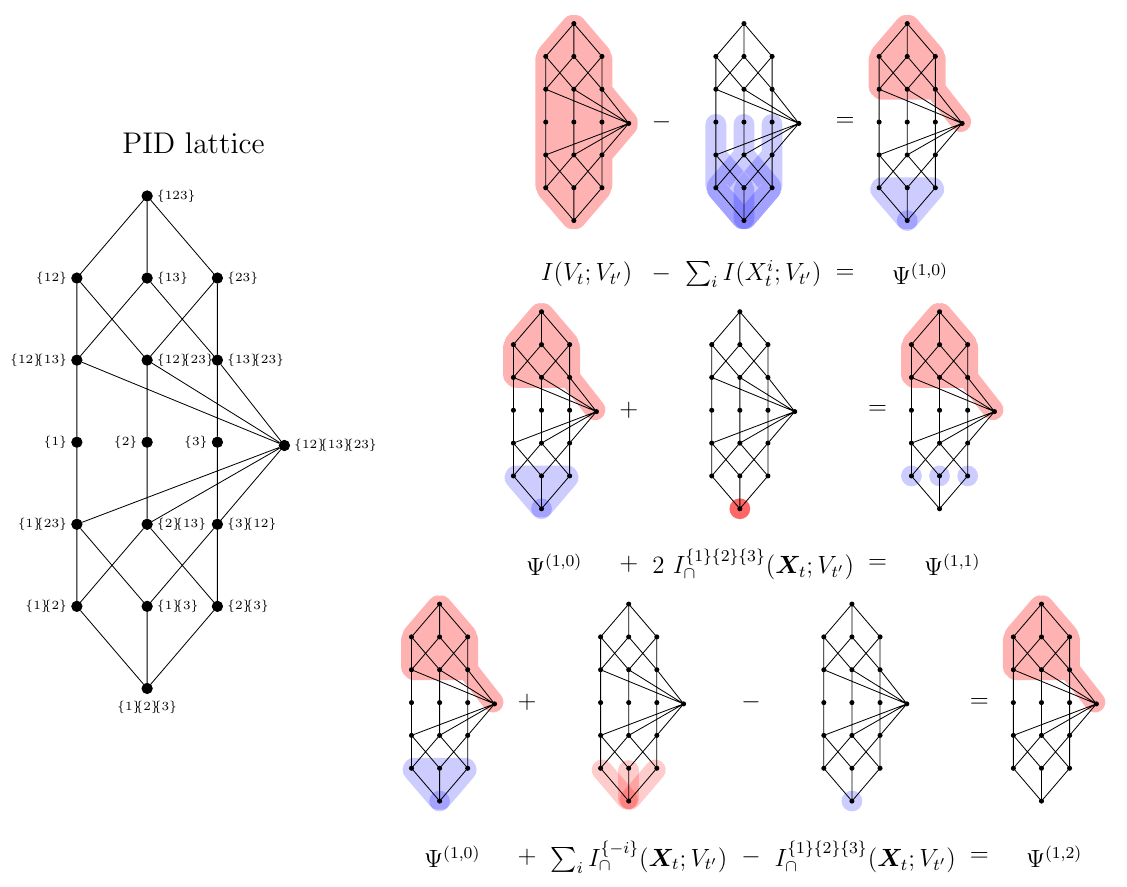}

  \caption{\textbf{Graphical illustration of lattice expansion for a system of $n = 3$ variables}. 
  Lattice nodes representing information atoms that get added together are marked in red, while the ones subtracted are marked in blue. Intuitively, the information provided by each of the three micro variables towards the macro includes the redundant information they share. When summing these mutual information quantities into a ``sum-of-parts" term that gets subtracted from the ``whole'', this redundancy is included three times rather than once. The first-order lattice expansion re-adds the redundancy shared between all components, while the second-order expansion corrects for redundant information among groups of two. In a system of three variables, only two orders of expansion are necessary.}

  \label{fig:lattice}
\end{figure*}

Intuitively, subtracting the mutual information from each source to the target will also subtract the redundant information they all provide together or as sub-groups of varying sizes.
To better illustrate this, we present the example of three source variables, visualised with PID lattices, in Fig.~\ref{fig:lattice}.  Recall redundant information in each node in the lattice, $\IR{\bm\alpha}$, includes the partial information provided by all the nodes $\bm\beta$, with $\bm\beta\preceq\bm\alpha$, below it. 
For $k=1$ and $q=0$, we subtract all atoms below and including \ensuremath{\{1\}, \{2\}, \{3\}}.
By inspecting the lattice, terms $\{1\}\{23\}$, $\{2\}\{13\}$ and $\{3\}\{12\}$ contain the redundant information of all atoms below. 
Thus when subtracting these terms, the redundancy of atoms one level below ($\{1\}\{2\}$, $\{1\}\{3\}$ and $\{2\}\{3\}$, with two singleton sets) is counted twice, so double-counted once, and the redundancy two levels below ($\{1\}\{2\}\{3\}$, with three singleton sets) is counted three times, so double-counted twice.

More generally, each subtracted term from the bottom half of the lattice is counted one time more than the term above. 
This is equivalent to the observation that any given atom $\bm\alpha$ in the lower half of the lattice, when containing, say, $k$ \textit{singleton} sets composed of a single variable (such as $\{1\}$), would be counted by a factor equal to the number of singleton sets minus 1, i.e.~$k-1$ times.

To counteract this double-counting and obtain a finer-grained measure, the main idea is to start re-adding redundancies from the bottom of the lattice to the whole-minus-sum estimator.

\subsection{A family of corrected measures}

\noindent Building on the ideas outlined above, we now outline a family of
more sensitive estimators obtained through a technique we call the
\emph{lattice expansion}. Intuitively, lattice expansion is a broadly applicable technique to refine Shannon-invariant measures, where we start from the coarse ``whole-minus-sum'' measures and then progressively add back the redundancy information atoms by expanding the redundancy lattice. Here, we illustrate the lattice expansion for the criterion for causal emergence $\Psi$.

The goal is to define a family of measures $\Psi^{(k,q)}$, where $q
= 0, \dots, n-1$ is the approximation order, such that $\Psi^{(k,0)} =
\Psi^{(k)}$ and $\Psi^{(1,n-1)}$ are equivalent to the first-order synergy in the system (see Appendix for more details and proofs). 

Knowing intuitively the source of the double-counted terms, we use information lattices to visualise the partial ordering between information atoms and formulate the full lattice expansion :
\begin{align}
  \Psi^{(k,q)} = \Psi^{(k)} + \sum_{r=n-q+1}^n \sum_{\substack{\bm\alpha\in\mathcal{M}^n\\|\bm\alpha| = r}} C^n_{q,r} \IR{\bm\alpha}(\bm X_t; V_{t'}) ~ ,
  \label{eq:expansion}
\end{align}
with the coefficients $C^n_{q,r}$ defined recursively, starting with  $C^n_{1,n} = n-1$:

\begin{align}
  C^n_{q,r} &= r-1 - \sum_{s=n-q+1}^{r-1} C^n_{q,s} \binom{r}{s} \\
  C^n_{q,n-q+1} &= n-q,
\end{align}
where $n$ is the system size, $q \in [1, n-1]$ the order of the expansion,
$r \in [n-q+1, n]$ the cardinality of the set of sources being considered, 
and $\binom{r}{s}$ the binomial operator, which returns the number of unordered subsets of size $r$ in a collection of $s$ variables.

As an example, the first and second lattice expansions yield a better estimation of $\Psi^{(k)}$:
\begin{align*}
  \Psi^{(k,1)} = \Psi^{(k,0)} & + (n-1) \IR{\{1\}...\{n\}}(\bm X_t; V_{t'})
\end{align*}
\begin{align*}
  \Psi^{(k,2)} = \Psi^{(k,0)} & + (n-2) \sum_{i=1}^n \IR{\{-i\}}(\bm X_t; V_{t'}) \\
                              & + \left(n-1 - n(n-2)\right) \IR{\{1\}...\{n\}}(\bm X_t; V_{t'})
\end{align*}
When computing $\Psi$, information shared by all $n$ components $\IR{\{1\}...\{n\}}(\bm X_t; V_{t'})$ -- the term at the bottom of the lattice -- is included in the sum $n$ times and as such, ``double-counted'' $n-1$ times. The first lattice expansion ($q=1$), will re-add this information.

Moving one level up the lattice, the information shared by sets of $n-1$ components $\bm X^{\{-i\}}$ is included in the sum $n-1$ times and double-counted $n-2$ times, so it is re-added in the second lattice expansion ($q=2$), on top of the redundancy from the first lattice expansion.
But each of these re-added sets contains information shared by $n-1$ components, and so, it includes the information shared by all components $\IR{\{1\}...\{n\}}(\bm X_t; V_{t'})$ $n-1$ times, double-counting it $n-2$ times for each of the $n$ sets in this re-added term, so it must be subtracted $n(n-2)$ times.

The number of evaluations of $\IR{\bm\alpha}$ grows as $\mathcal{O}(\binom{n}{q-1})$. 
Since all the redundancies involve only singleton
sources, these terms are also easy to estimate empirically, since the sampling
bias in redundancy is of the same order of magnitude as that of the MI of the
sources involved~\cite{venkatesh2023gaussian}.

Notably, using the lattice expansion does involve choosing a particular PID
measure from the multiple options
available~\cite{williams2010nonnegative,barrett2015exploration,ince2017measuring}.
It is worth noting that in practical analyses, these tend to yield
qualitatively similar results~\cite{tax2017partial,rosas2020operational}.
However, since information measures capture statistical interdependencies rather than causal mechanisms~\cite{Liardi2025SimplePS}, 
using different synergy or redundancy functions requires careful interpretation.

Having introduced the criterion, we also created an implementation in Python using the Java Information Dynamics Toolkit (JIDT)~\cite{lizier2014jidt}. The implementation of the lattice decomposition was tested against the decomposition provided by the Discrete Information Theory (DIT)~\cite{james2018dit} package. The implementation is available online~\cite{mis2024} and was used on the trajectories of collective motion, which will be presented in the following two case studies.

\section{Case studies}

We now present two case studies that illustrate the benefits of the proposed approach. For these analyses, 
we use the minimum mutual information (MMI) as the redundancy function, which is computationally efficient and was shown to be either equivalent or to behave qualitatively similarly to various other measures of redundancy in Gaussian systems \cite{barrett2015exploration}.

\subsection{Reynolds flocking model}

\noindent \textbf{Setting} \quad A quintessential example of emergence %
is the collective behaviour of animal groups in motion: flocks of birds, schools of fish, herds of herbivores and so on. 
To first of all relate the methods presented in the current paper to previous work, the 
first canonical example of emergence we study with the new estimators is the Reynolds flocking
model \cite{reynolds1987flocks}. 

One of the first models to realistically simulate the phenomenology of collective motion in flocks, schools and herds, the Reynolds model
defines a multi-agent system of ``boids" (or ``bird-oid objects") which move in a 2 or 3-dimensional space, and interact with their neighbours on a given radius $r$, following three different types of social forces:
\begin{itemize}
    \item \textit{Aggregation} ($a_1$): tendency to fly towards the centre of the flock
    \item \textit{Avoidance} ($a_2$): tendency to avoid hitting the nearest neighbour
    \item \textit{Alignment} ($a_3$): tendency to align direction with nearest neighbours on a radius $r$
\end{itemize}

The model is realistic to describe collective animal behaviour, insofar as there is a degree of momentum in their motion: forces are implemented as ``steerage" towards or away from the position or direction of others. 
Increasing the avoidance parameter ($a_2$) will encourage the boids to distance away from each other, and as such, will decrease the effect of the alignment force. Manipulating this parameter produces qualitatively very different behaviours: for low avoidance, the boids gather in a cyclone, known as milling behaviour; for high avoidance, the boids no longer fly together; but for an intermediary value, one can observe a chimeric behaviour where the conflicting tendencies between order and disorder create the adaptive and complex emergent behaviour we often see in nature. Snapshots of characteristic collective behaviour are in Fig.~\ref{fig:phenom-reynolds}.

\begin{figure*}[t]
    \centering
    \includegraphics[width=0.32\linewidth]{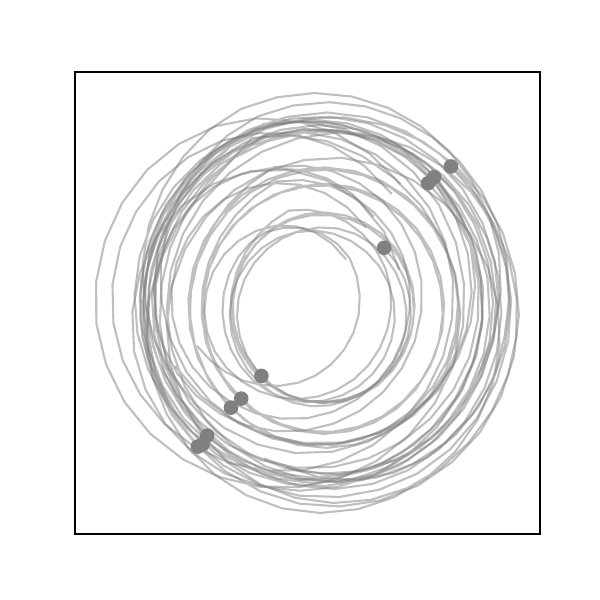}
    \includegraphics[width=0.32\linewidth]{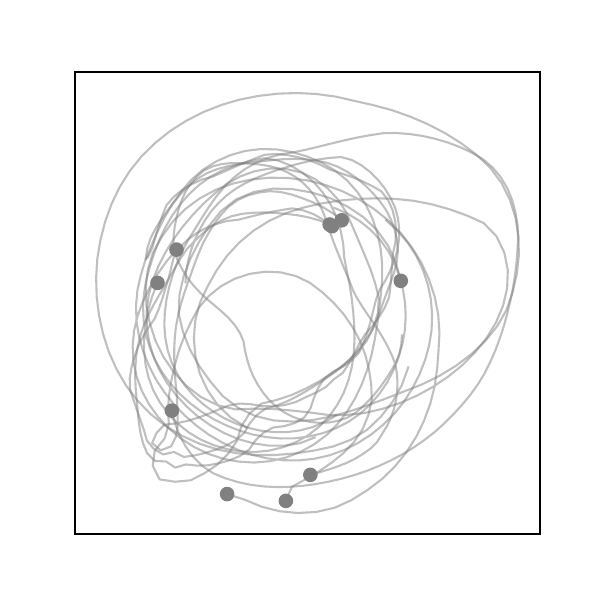}
    \includegraphics[width=0.32\linewidth]{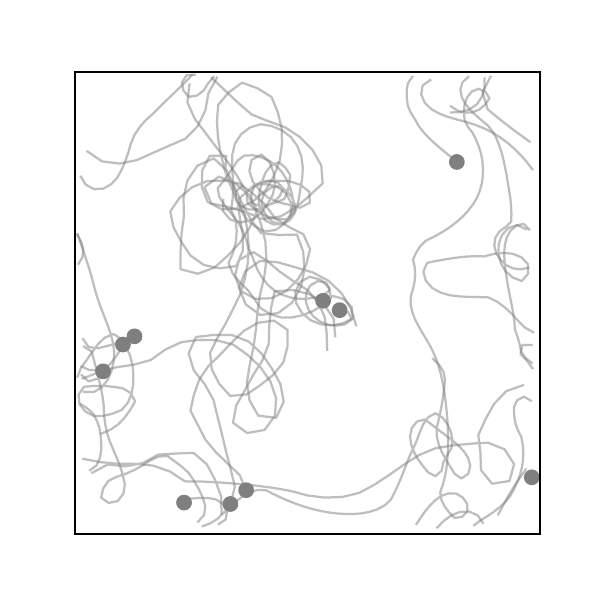}
\caption{\textbf{Phenomenonologically different instances of the Reynolds flocking model under increasing avoidance.} 
    (a) No avoidance ($a_2 = 0$): shows rigid milling behaviour, which should manifest very high redundancy. 
    (b) Intermediary avoidance ($a_2 = 0.1$): reminds of a flock of birds flying around their nest, or fish in a tank. This behaviour should be characterised by a balance of redundant and synergetic information. 
    (c) Higher avoidance ($a_2 = 0.2$): almost random behaviour, which should be characterised by small redundancy and small synergy. }
    \label{fig:phenom-reynolds}
\end{figure*}

\vspace{1em}
\noindent \textbf{Methods} \quad
For the sake of reproducibility, the simulations were configured with exactly the same parameters and random seeds as in~\cite{rosas2020reconciling}: $N = 10$ boids were simulated on a torus of side length $L = 200$, initialised with the same velocity, with random initial orientations and positions; the interaction radius was taken as $r=20$, a tenth of the total space. The velocity is the same for all boids, and the aggregation ($a_1 = 0.15$) and alignment ($a_3=0.25$) parameters were kept fixed, while the avoidance parameter was varied ($a_2 \in \{ 0.0, 0.1, 0.2 \}$) to create more or less `rigid' flocks. $R=20$ simulations were run for each parameter set for statistical robustness. 

The $\Psi$ measure was computed using the individual trajectories as the micro components of the system, while the group's centre of mass was used as a macro candidate property for an emergent feature. The time difference between the two time series was set as $t'-t=1$.

Before estimating the three quantities and applying the emergence criteria to the data, the individual positions were preprocessed by computing the distance from the centre of mass and first-order differentiating the resulting time series, following~\cite{seth2010measuring,rosas2020reconciling}.
This method reduces the underlying shared information between the velocities, as well as the positions, which are constrained within finite spatial coordinates. Calculating the distance to the centre of mass removes some of the shared information inherent in the positions. 
This preprocessing also removes some difficulties of estimating probability distributions of positions in a space with periodic boundaries (as seen, for example, in the snapshot in Fig.~\ref{fig:phenom-reynolds} (c)), which require circular statistics. Redundancy was computed by the Minimum mutual information (MMI) \cite{barrett2015exploration}.

As before, we hypothesised that the intermediary configuration ($a_2=0.1$) manifests the highest $\Psi$ due to its high synergy, but also that the lowest avoidance configuration ($a_2=0$) will be characterised by higher redundancy than the others, due to the similarity in the motion of boids. 
While this point was suggested in previous work, namely that the low avoidance scenario is dominated not by a reduction in synergy, but by an increase in redundancy, which effectively increases the synergy threshold needed for a positive $\Psi$ \cite{rosas2020reconciling},
we can now observe this directly in our results.
The higher the redundancy, the higher the threshold for synergy.

\begin{figure*}[ht]
\begin{minipage}[t]{0.48\linewidth}
    \centering
    \includegraphics[width=0.88\linewidth]{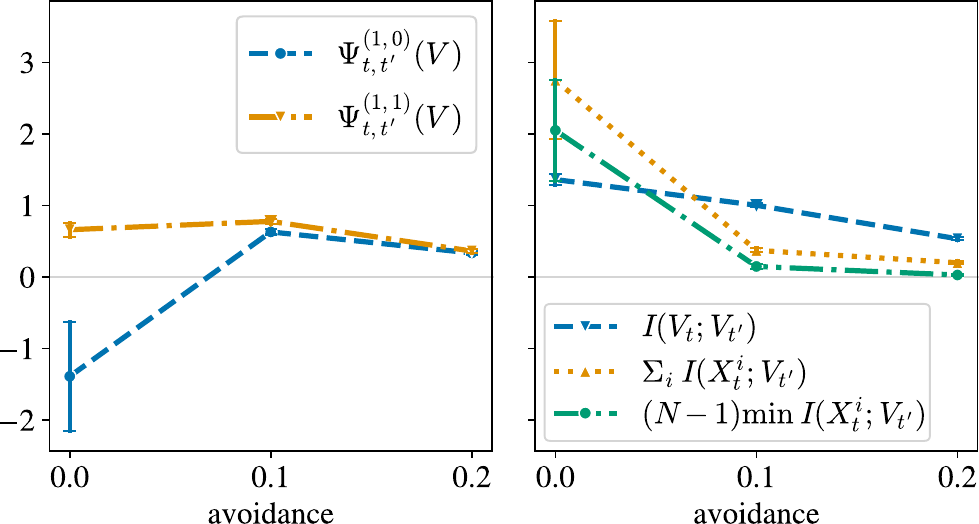}
    \caption{\textbf{ Causal emergence and first-order redundancy in the flocking boids model.}
$\Psi^{(1,0)}$ peaks in the intermediary state, and is negative in the low avoidance scenario. But when considering the first-order redundancy in the system in computing $\Psi^{(1,1)}$, the criterion also correctly finds the no-avoidance, redundancy-driven milling movement as emergent. Error bars show standard error across $R=20$ simulations for each parameter set.}
    \label{fig:results-reynolds}
\end{minipage}
\hspace{0.5cm}
\begin{minipage}[t]{0.48\linewidth}
    \centering
    \includegraphics[width=\linewidth]{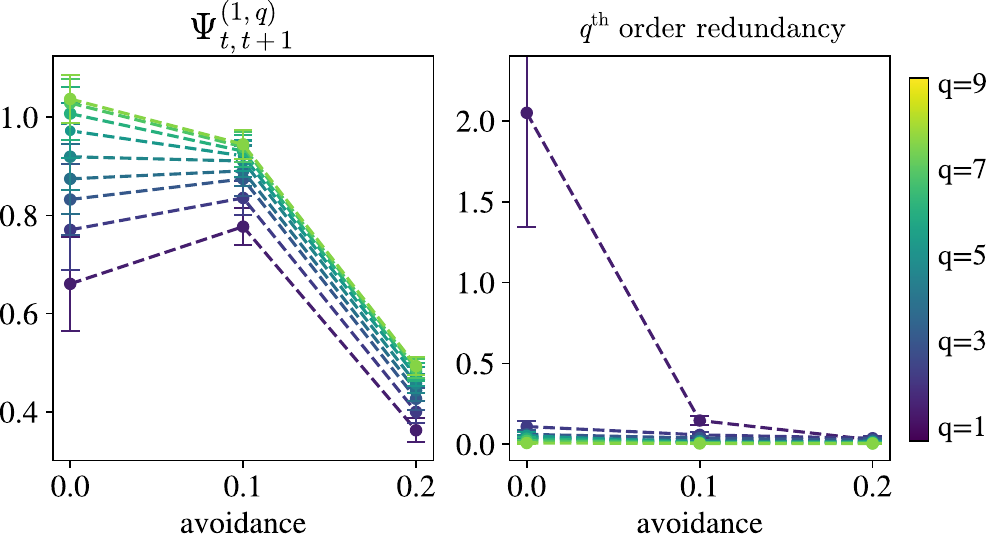}
    \caption{\textbf{ Causal emergence and redundancy in the flocking boids model.}
    Only the redundancy in the first lattice expansion is relevant in this system, being orders of magnitude while higher expansion orders ($q=2$ to $q=10$) are negligible. Error bars represent standard error across $R=20$ simulations for each parameter set.
    }
    \label{fig:results-reynolds-red}
\end{minipage}
\end{figure*}

\vspace{1em}
\noindent \textbf{Results} \quad
Having computed the amount of redundancy in the system (Fig.~\ref{fig:results-reynolds}), we find that the high avoidance case (with $\tilde\Psi^{(1,0)}=0.335, SD=0.046$) is barely affected by the first lattice expansion, while the intermediary regime is only affected slightly (with $\tilde\Psi^{(1,0)}=0.63, SD=0.078$ and a finer-grained $\tilde\Psi^{(1,1)}=0.777, SD=0.082$). The criteria would consider both these systems emergent; however, the average value of $\Psi$ is considerably smaller in the chaotic than in the intermediary regime.

However, in the no-avoidance, rigid case, which yields a negative value on average ($\tilde\Psi^{(1,0)}=-1.39, SD=1.653$), a more accurate $\Psi$ using the $q=1$ redundancy gives a positive value ($\tilde\Psi^{(1,1)}=0.66, SD=0.208$), suggesting that this system is one such case of emergence where higher redundancy hinders estimation of $\Psi$.

Most interestingly, the complete lattice expansion for $\Psi$, with $q=n-1$, reveals an unexpected result: namely a higher emergence value for the rigid case than the intermediary regime, which turns the $\Psi$ criterion from a function which peaks in the interdmediary regime to one that is monotonically decreasing with reducing the collective behaviour. Crucially, this is unlike synergy-redundancy indices in general.
However, the no-avoidance case is an example phenomenon where emergence is driven by redundancy: the full lattice expansion $\Psi^{(1,n-1)}$ may be used to detect emergence in strongly redundant systems and not just strongly synergistic ones. A more fine-grained exploration of simulation parameters, as well as different kinds of movement, should be explored before drawing a strict conclusion.

It is interesting to consider whether higher-order redundancy exists in the system, i.e. whether pairs or larger groups of boids contribute redundant information. Our results suggest that it is generally not the case (Fig.~\ref{fig:results-reynolds-red}). 
This is in accordance with the nature of the model and its phenomenology: all boids are identical and follow the same distributions, and on average interact the same way with all neighbours. In other words, the interactions between these boids and as such the information they share are homogenous, while in real-world complex systems manifesting emergence, the agents are likely to interact in more heterogeneous ways, especially when studying dynamics of certain complex collective behaviour.
As such, it is interesting to compare the results obtained in the analysis of a flocking model with real-world data of collective motion in nature.

\subsection{Schooling fish}

\begin{figure*}[th]
    \centering
    \includegraphics[width=0.96\linewidth]{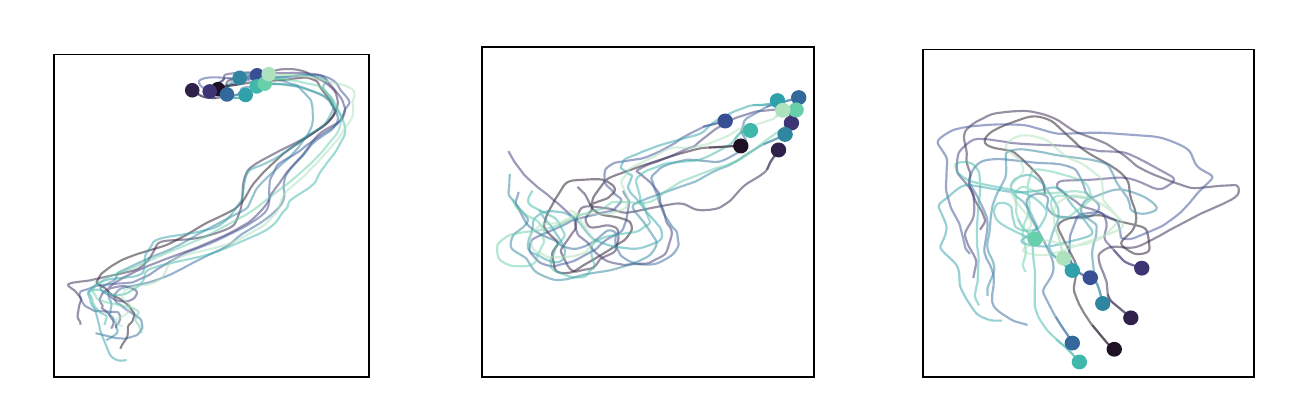}
    \caption{\textbf{Trajectories of 200 timesteps (24 seconds) of ten schooling fish from three different measurements}, showing increasing degrees of dispersion, or decreasing degrees of aggregation, while maintaining phenomenologically similar dynamical behaviour, which reminds of the intermediary regime in the Reynolds model (Fig.~\ref{fig:phenom-reynolds}).}
    \label{fig:phenom-fish}
\end{figure*}

\vspace{1em}
\noindent \textbf{Setting} \quad
\noindent As a secondary case study, we look at small schools of $ayu$ (Japanese sweetfish) swimming in a tank.
The collective behaviour shown by the schooling fish is easily visible with the naked eye (Fig.~\ref{fig:phenom-fish}). The behaviours persist for the entire time series, but the fish manifest different degrees of aggregation, from swimming very close to each other and maintaining mostly the same directions, resulting in a more rigid formation with tightly-knit trajectories, to swimming farther away and sometimes abruptly changing direction, resulting in an almost swarm-like milling at certain points.

A known problem in collective behaviour research is that, while present in the case of mathematical modelling, the mechanisms of neighbour interactions are notoriously difficult to estimate and ascertain in real data. As such, we cannot rely on any knowledge of micro interactions to classify the school's behaviour into a specific regime or phase. 

Nevertheless, research in swarm behaviour has shown that swarming, schooling and flocking often emerge thanks to the conflict between individual and collective tendencies, which could be conceptually linked to the intermediary regime in the Reynolds model. Moreover, one can observe the similarity in phenomenology between the real-world data and the intermediary regime of the simulation with the naked eye, allowing us to hypothesise that the two systems, both of size $N=10$, would show similar values for the emergence criteria.

\vspace{1em}
\noindent \textbf{Methods} \quad
The trajectory data for the schooling fish was obtained experimentally by Niizato et al.~\cite{niizato2024fish}. $N=10$ fish were placed in a 33m$^2$ tank of 15cm deep, and filmed from above to obtain trajectories that appear two-dimensional. The fish were tracked for 8 to 12 min at 20 frames per second using computer vision, producing time series around $T=10000$ in length. From the data made publicly available by the research team, we selected the first $T=10000$ timesteps in six of their data sets. Please see the original paper for further experimental details.

We estimate the $\Psi$ measure in the same way as for the Reynolds model, by choosing the centre of mass of the school as the macro feature $V$, then pre-processing all spatial coordinates into differentiated distances from the centre of the environment. 

\begin{figure*}[th]
    \centering
    \includegraphics[width=0.8\linewidth]{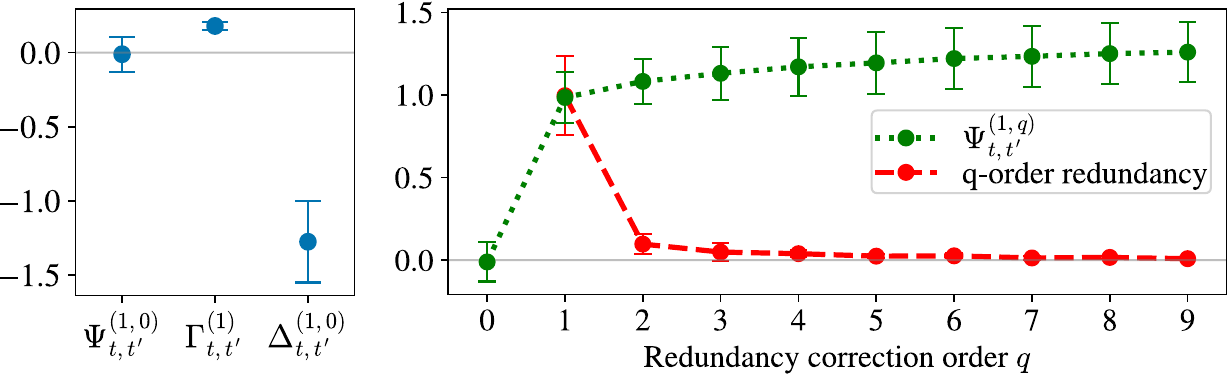}
    \caption{\textbf{Average causal emergence and redundancy in fish schools.}
    (a) Average quantities  for causal emergence ($\Psi$), causal decoupling ($\Gamma$), and downwards causation ($\Delta$) measures with $q=0$. According to the criteria in~\cite{rosas2020reconciling}, the schools manifest neither causal emergence (i.e. $\Psi > 0$) nor downwards causation (i.e. $\Delta > 0$) when double-counted shared information in the system is not accounted for.
    This is due to the high redundancy in the system, with the fish's trajectories being very similar. 
    (b) $\Psi$ and redundancy quantities computed by lattice expansion for $q \in \{0, 1...9 \}$, averaged across the six datasets. $\Psi$ is around 0 unless at least the first lattice expansion is applied, due to high redundancy in the system, corresponding to the shared information among all 10 fish. This is considerably higher than the average shared information in smaller subgroups, which supports the idea that schooling behaviour is egalitarian, with all interactions being similar to each-other.
    }
    \label{fig:results-fish-psi-avg}
\end{figure*}

\vspace{1em}
\noindent \textbf{Results} \quad
We averaged our $\Psi$ results across all six groups of fish, expecting to observe $\Psi > 0$ for all groups. 
But instead, we found that the emergence criteria with $q=1$ show schooling fish manifest neither causal emergence ($\tilde\Psi^{(1,0)}=-0.01, SD=0.14$), nor downwards causation ($\tilde\Delta^{(1,0)}=-1.28, SD=0.34$).

Our results show $\Psi$ is around 0 unless a first-order or higher lattice expansion is applied, due to the high amount of first-order redundancy in the system: namely, all sources $X_i$ at time $t$ provide at least that much information to the target $V$ at time $t+1$. 
We already see the need for the lattice expansion, as, in spite of phenomenologically emergent behaviour, $\Psi^{(1,0)}\approx0$ without re-adding redundancies, but all $q\geq 1$ quantities are clearly positive, with $\Psi^{(1,1)}=0.98, SD=0.19$ suggesting the emergence of an aggregated group around the centre of mass, and $\Delta^{(1,1)}_{t,t'}=-1.10, SD=0.25$ suggesting there is no downwards causation.

Higher-order redundancy then decreases abruptly, similar to both the high-rigidity case and the intermediary regime in the Reynolds model.  
In the school of fish, this may have a natural interpretation: very recently, it has been shown using neuroimaging that fish tend to follow a small number of neighbours when they swim (one or two)~\cite{oscar_vr_fish_2023}, but that they are very sensitive to changes in behaviour on their perception radius which cause very quick behavioural cascades, opening up the possibility of spontaneous coordinated behaviour amongst all of them as opposed to within smaller sub-groups~\cite{mugica_scale-free_2022}. 

The fine-grained estimation of redundancy across groups of components of various sizes, or, in other words, the decay in redundancy for increasing  $q$,  opens up possible new directions in the quantification and understanding of the elusive social interactions between animals, which give rise to the marvels of swarm intelligence seen in flocking, schooling and herding behaviour.

\section{Discussion}

\noindent In this work, we introduced a finer-grained decomposition of Shannon-invariant estimators of emergence that allows us to explicitly account for redundant information, and thereby refine the effective synergy threshold required to detect emergence. By expressing the emergence criterion in terms of contributions from different orders of redundancy, our approach can systematically correct for double-counted shared information and thus reduce false negatives often encountered in highly ordered systems. 

Conceptually, a redundancy correction based on the lattice expansion of order $q$ captures the information (here, quantified using minimum mutual information) that is shared by groups of $n-q+1$ sources about a given target. Applying this framework to Reynolds flocking dynamics, using the centre of mass as a macroscopic coarse-grained variable, we confirmed the intuition that strongly ordered configurations with low avoidance exhibit the largest redundancy. Incorporating a first-order redundancy correction was found to be enough to reveal causal emergence in this highly ordered regime, where the uncorrected estimator can otherwise fail to detect it.

We then applied the same methodology to empirical trajectory data from schooling fish. In these experiments, uncorrected emergence metrics fail to identify causal emergence. However, once first-order redundancy is taken into account, causal emergence becomes evident. Moreover, examining the decay of higher-order redundancy reveals additional structure in how information is shared among individuals, hinting at non-trivial patterns of interaction that go beyond simple pairwise organisation.

Taken together, these results demonstrate that redundancy-aware refinements of emergence measures can substantially improve sensitivity in both simulated and real-world systems. Beyond detecting emergence, the detailed profile of redundancy across interaction orders provides a promising avenue for probing the organisation and information-sharing architecture underlying collective behaviour in natural and artificial systems. 
These new tools open various doors for the investigation of large complex systems which display both redundancy and emergence in different places for different purposes --- collective intelligence, swarm behaviour, and last, but not least, the human brain~\cite{luppi2020synergistic,luppi2024information}.

\section*{Acknowledgements}

\noindent F.R. was supported by the UK ARIA Safeguarded AI programme, the PIBBSS
Affiliateship programme, and Open Philanthropy. H.R. is supported by the Eric and Wendy Schmidt AI for Science Fellowship. 
D.B. is funded by the Wellcome Trust (grant no. 210920/Z/18/Z).
H.J.J. and H.R. were also supported by the Statistical Physics of Cognition Project funded by the EPSRC (Grant No. EP/W024020/1).

\bibliographystyle{IEEEtran}
\bibliography{main}

@article{shannon_1953,
	title        = {The lattice theory of information},
	author       = {Shannon, C.},
	year         = 1953,
	month        = {Feb},
	journal      = {Transactions of the IRE Professional Group on Information Theory},
	volume       = 1,
	number       = 1,
	pages        = {105–107},
	doi          = {10.1109/tit.1953.1188572}
}

@article{lizier2014jidt,
  title={{JIDT: An} information-theoretic toolkit for studying the dynamics of complex systems},
  author={Lizier, Joseph T},
  journal={Frontiers in Robotics and AI},
  volume={1},
  pages={11},
  year={2014},
  publisher={Frontiers}
}

@article{luppi2024information,
  title={Information decomposition and the informational architecture of the brain},
  author={Luppi, Andrea I and Rosas, Fernando E and Mediano, Pedro AM and Menon, David K and Stamatakis, Emmanuel A},
  journal={Trends in Cognitive Sciences},
  volume={28},
  number={4},
  pages={352--368},
  year={2024},
  publisher={Elsevier}
}

@book{reynolds1987flocks,
  title={Flocks, Herds and Schools: {A} Distributed Behavioral Model},
  author={Reynolds, Craig W},
  volume={21},
  number={4},
  year={1987},
  publisher={ACM}
}

@article{seth2010measuring,
  title={Measuring autonomy and emergence via Granger causality},
  author={Seth, Anil K},
  journal={Artificial Life},
  volume={16},
  number={2},
  pages={179--196},
  year={2010},
  publisher={MIT Press}
}

@article{williams2010nonnegative,
  title={Nonnegative decomposition of multivariate information},
  author={Williams, Paul L and Beer, Randall D},
  journal={arXiv preprint arXiv:1004.2515},
  year={2010}
}

@article{james2018dit,
  title={``dit``: {A} Python package for discrete information theory},
  author={James, Ryan G and Ellison, Christopher J and Crutchfield, James P},
  journal={Journal of Open Source Software},
  volume={3},
  number={25},
  pages={738},
  year={2018}
}

@article{luppi2020synergistic,
  title={A synergistic core for human brain evolution and cognition},
  author={Luppi, Andrea I and Mediano, Pedro AM and Rosas, Fernando E and Holland, Negin and Fryer, Tim D and O'Brien, John T and Rowe, James B and Menon, David K and Bor, Daniel and Stamatakis, Emmanuel A},
  journal={BioRxiv},
  year={2020},
  publisher={Cold Spring Harbor Laboratory}
}

@article{klein2020emergence,
  title={The emergence of informative higher scales in complex networks},
  author={Klein, Brennan and Hoel, Erik},
  journal={Complexity},
  volume={2020},
  year={2020},
  publisher={Hindawi}
}

@article{barrett2015exploration,
  title={Exploration of synergistic and redundant information sharing in static and dynamical Gaussian systems},
  author={Barrett, Adam B},
  journal={Physical Review E},
  volume={91},
  number={5},
  pages={052802},
  year={2015},
  publisher={APS}
}

@article{mediano2019beyond,
  title={Beyond integrated information: A taxonomy of information dynamics phenomena},
  author={Mediano, Pedro AM and Rosas, Fernando and Carhart-Harris, Robin L and Seth, Anil K and Barrett, Adam B},
  journal={arXiv preprint arXiv:1909.02297},
  year={2019}
}

@article{rosas2020reconciling,
  title={Reconciling emergences: An information-theoretic approach to identify causal emergence in multivariate data},
  author={Rosas, Fernando E and Mediano, Pedro AM and Jensen, Henrik J and Seth, Anil K and Barrett, Adam B and Carhart-Harris, Robin L and Bor, Daniel},
  journal={PLoS Computational Biology},
  volume={16},
  number={12},
  pages={e1008289},
  year={2020},
  publisher={Public Library of Science San Francisco, CA USA}
}

@article{barnett2021dynamical,
  title={Dynamical independence: Discovering emergent macroscopic processes in complex dynamical systems},
  author={Barnett, Lionel and Seth, Anil K},
  journal={arXiv preprint arXiv:2106.06511},
  year={2021}
}

@book{lizier2012local,
  title={The Local Information Dynamics of Distributed Computation in Complex Systems},
  author={Lizier, Joseph T},
  year={2012},
  publisher={Springer Science \& Business Media}
}

@article{varley2022emergence,
  title={Emergence as the conversion of information: a unifying theory},
  author={Varley, Thomas F and Hoel, Erik},
  journal={Philosophical Transactions of the Royal Society A},
  volume={380},
  number={2227},
  pages={20210150},
  year={2022},
  publisher={The Royal Society}
}

@article{luppi2021reduced,
    title={Reduced causal emergence in the brains of chronically unconscious patients},
    author={Andrea I Luppi and Pedro A. M. Mediano and Fernando E. Rosas and Judith Allanson and John D. Pickard and Guy B. Williams and Michael M. Craig and Paola Finoia and Alexander R. D. Peattie and Peter Coppola and David K. Menon and Daniel Bor and Emmanuel A. Stamatakis},
    journal={In Prep.},
    year={2021}
}

@article{rosas2019quantifying,
  title = {Quantifying high-order interdependencies via multivariate extensions of the mutual information},
  author = {Rosas, Fernando E. and Mediano, Pedro A. M. and Gastpar, Michael and Jensen, Henrik J.},
  journal = {Physical Review E},
  volume = {100},
  issue = {3},
  pages = {032305},
  numpages = {15},
  year = {2019},
  month = {Sep},
  publisher = {American Physical Society},
}

@inproceedings{gat1999synergy,
	Author = {Gat, Itay and Tishby, Naftali},
	Booktitle = {Advances in Neural Information Processing Systems},
	Pages = {111--117},
	Title = {Synergy and redundancy among brain cells of behaving monkeys},
	Year = {1999}
}

@article{tax2017partial,
  AUTHOR = {Tax, Tycho M.S. and Mediano, Pedro A. M. and Shanahan, Murray},
  TITLE = {The Partial Information Decomposition of Generative Neural Network Models},
  JOURNAL = {Entropy},
  VOLUME = {19},
  YEAR = {2017},
  NUMBER = {9},
  ARTICLE-NUMBER = {474},
  ISSN = {1099-4300},
  DOI = {10.3390/e19090474}
}

@article{rosas2020operational,
  title={An operational information decomposition via synergistic disclosure},
  author={Rosas, Fernando E and Mediano, Pedro A M and Rassouli, Borzoo and Barrett, Adam B},
  journal={Journal of Physics A: Mathematical and Theoretical},
  doi={10.1088/1751-8121/abb723},
  volume={53},
  number={48},
  pages={485001},
  year={2020},
  publisher={IOP Publishing}
}

@article{niizato2024fish, title={Information structure of heterogeneous criticality in a fish school}, DOI={https://doi.org/10.1101/2024.02.18.578833}, journal={bioRxiv (Cold Spring Harbor Laboratory)}, publisher={Cold Spring Harbor Laboratory}, author={Takayuki Niizato and Sakamoto, Kotaro and Yohichi Mototake and Murakami, Hisashi and Takenori Tomaru}, year={2024}, month={Feb} }

@book{jensen2022complexity,
  title={Complexity science: the study of emergence},
  author={Jensen, Henrik Jeldtoft},
  year={2022},
  publisher={Cambridge University Press}
}

@article{anderson1972more,
  title={More Is Different: Broken symmetry and the nature of the hierarchical structure of science.},
  author={Anderson, Philip W},
  journal={Science},
  volume={177},
  number={4047},
  pages={393--396},
  year={1972},
  publisher={American Association for the Advancement of Science}
}

@article{durant2019role,
  title={The role of early bioelectric signals in the regeneration of planarian anterior/posterior polarity},
  author={Durant, Fallon and Bischof, Johanna and Fields, Chris and Morokuma, Junji and LaPalme, Joshua and Hoi, Alison and Levin, Michael},
  journal={Biophysical journal},
  volume={116},
  number={5},
  pages={948--961},
  year={2019},
  publisher={Elsevier}
}

@inproceedings{butterfield2012emergence,
  title={Emergence and reduction combined in phase transitions},
  author={Butterfield, Jeremy and Bouatta, Nazim},
  booktitle={AIP Conference Proceedings 11},
  volume={1446},
  number={1},
  pages={383--403},
  year={2012},
  organization={American Institute of Physics}
}

@inproceedings{mitchell1996evolving,
  title={Evolving cellular automata with genetic algorithms: A review of recent work},
  author={Mitchell, Melanie and Crutchfield, James P and Das, Rajarshi and others},
  booktitle={Proceedings of the First international conference on evolutionary computation and its applications (EvCA’96)},
  volume={8},
  pages={23--30},
  year={1996},
  organization={Moscow}
}

@article{schelling1971dynamic,
  title={Dynamic models of segregation},
  author={Schelling, Thomas C},
  journal={Journal of mathematical sociology},
  volume={1},
  number={2},
  pages={143--186},
  year={1971},
  publisher={Taylor \& Francis}
}

@book{broad1925,
  title={The Mind and Its Place in Nature},
  author={Charlie Dunbar Broad},
  year={1925},
  publisher={K. Paul, Trench, Trubner \& Co., Ltd.},
  address={London}
}

@article{cohen2007explaining,
  title={Explaining a complex living system: dynamics, multi-scaling and emergence},
  author={Cohen, Irun R and Harel, David},
  journal={Journal of the royal society interface},
  volume={4},
  number={13},
  pages={175--182},
  year={2007},
  publisher={The Royal Society London}
}

@article{corning2012re,
  title={The re-emergence of emergence, and the causal role of synergy in emergent evolution},
  author={Corning, Peter A},
  journal={Synthese},
  volume={185},
  number={2},
  pages={295--317},
  year={2012},
  publisher={Springer}
}

@article{krakauer2002redundancy,
  title={Redundancy, antiredundancy, and the robustness of genomes},
  author={Krakauer, David C and Plotkin, Joshua B},
  journal={Proceedings of the National Academy of Sciences},
  volume={99},
  number={3},
  pages={1405--1409},
  year={2002},
  publisher={The National Academy of Sciences}
}

@misc{lizier2018information,
  title={Information decomposition of target effects from multi-source interactions: Perspectives on previous, current and future work},
  author={Lizier, Joseph T and Bertschinger, Nils and Jost, J{\"u}rgen and Wibral, Michael},
  journal={Entropy},
  volume={20},
  number={4},
  pages={307},
  year={2018},
  publisher={MDPI}
}

@article{harder2013bivariate,
  title={Bivariate measure of redundant information},
  author={Harder, Malte and Salge, Christoph and Polani, Daniel},
  journal={Physical Review E—Statistical, Nonlinear, and Soft Matter Physics},
  volume={87},
  number={1},
  pages={012130},
  year={2013},
  publisher={APS}
}

@incollection{griffith2014quantifying,
  title={Quantifying synergistic mutual information},
  author={Griffith, Virgil and Koch, Christof},
  booktitle={Guided self-organization: inception},
  pages={159--190},
  year={2014},
  publisher={Springer}
}

@article{bertschinger2014quantifying,
  title={Quantifying unique information},
  author={Bertschinger, Nils and Rauh, Johannes and Olbrich, Eckehard and Jost, J{\"u}rgen and Ay, Nihat},
  journal={Entropy},
  volume={16},
  number={4},
  pages={2161--2183},
  year={2014},
  publisher={Multidisciplinary Digital Publishing Institute}
}

@article{ince2017measuring,
  title={Measuring multivariate redundant information with pointwise common change in surprisal},
  author={Ince, Robin AA},
  journal={Entropy},
  volume={19},
  number={7},
  pages={318},
  year={2017},
  publisher={MDPI}
}

@article{james2018unique,
  title={Unique information via dependency constraints},
  author={James, Ryan G and Emenheiser, Jeffrey and Crutchfield, James P},
  journal={Journal of Physics A: Mathematical and Theoretical},
  volume={52},
  number={1},
  pages={014002},
  year={2018},
  publisher={IOP Publishing}
}

@article{oscar_vr_fish_2023,
  title     = {A simple cognitive model explains movement decisions in zebrafish while following leaders},
  author    = {Oscar, Lital and Li, Liang and Gorbonos, Dan and Iain Couzin and Gov, Nir S},
  year      = 2023,
  month     = may,
  journal   = {Physical Biology},
  publisher = {IOP Publishing},
  volume    = 20,
  number    = 4,
  pages     = {045002–045002},
  doi       = {10.1088/1478-3975/acd298},
}

@article{rosas2025characterising,
  title={Characterising high-order interdependence via entropic conjugation},
  author={Rosas, Fernando E and Gutknecht, Aaron J and Mediano, Pedro AM and Gastpar, Michael},
  journal={Communications Physics},
  volume={8},
  number={1},
  pages={347},
  year={2025},
  publisher={Nature Publishing Group UK London}
}

@article{rosas2024software,
  title={Software in the natural world: A computational approach to hierarchical emergence},
  author={Rosas, Fernando E and Geiger, Bernhard C and Luppi, Andrea I and Seth, Anil K and Polani, Daniel and Gastpar, Michael and Mediano, Pedro AM},
  journal={arXiv preprint arXiv:2402.09090},
  year={2024}
}

@article{yuan2024emergence,
  title={Emergence and causality in complex systems: a survey of causal emergence and related quantitative studies},
  author={Yuan, Bing and Zhang, Jiang and Lyu, Aobo and Wu, Jiayun and Wang, Zhipeng and Yang, Mingzhe and Liu, Kaiwei and Mou, Muyun and Cui, Peng},
  journal={Entropy},
  volume={26},
  number={2},
  pages={108},
  year={2024},
  publisher={MDPI}
}

@article{mcsharry2024learning,
  title={Learning diverse causally emergent representations from time series data},
  author={McSharry, David and Kaplanis, Christos and Rosas, Fernando E and Mediano, Pedro A},
  journal={Advances in Neural Information Processing Systems},
  volume={37},
  pages={119547--119572},
  year={2024}
}

@article{tolle2024evolving,
  title={Evolving reservoir computers reveals bidirectional coupling between predictive power and emergent dynamics},
  author={Tolle, Hanna M and Luppi, Andrea I and Seth, Anil K and Mediano, Pedro AM},
  journal={arXiv preprint arXiv:2406.19201},
  year={2024}
}

@article{pigozzi2025associative,
  title={Associative conditioning in gene regulatory network models increases integrative causal emergence},
  author={Pigozzi, Federico and Goldstein, Adam and Levin, Michael},
  journal={Communications Biology},
  volume={8},
  number={1},
  pages={1027},
  year={2025},
  publisher={Nature Publishing Group UK London}
}

@article{mugica_scale-free_2022,
  title={Scale-free behavioral cascades and effective leadership in schooling fish},
  author={M{\'u}gica, Julia and Torrents, Jordi and Crist{\'\i}n, Javier and Puy, Andreu and Miguel, M Carmen and Pastor-Satorras, Romualdo},
  journal={Scientific reports},
  volume={12},
  number={1},
  pages={10783},
  year={2022},
  publisher={Nature Publishing Group UK London}
}

@misc{mis2024, title={Reconciling Emergences}, 
url={https://github.com/mearlboro/reconciling-emergences}, 
journal={GitHub}, author={Madalina I. Sas}, year={2024}, month={Apr} }

@article{mediano2022greater,
  title={Greater than the parts: {A} review of the information decomposition approach to causal emergence},
  author={Mediano, Pedro AM and Rosas, Fernando E and Luppi, Andrea I and Jensen, Henrik J and Seth, Anil K and Barrett, Adam B and Carhart-Harris, Robin L and Bor, Daniel},
  journal={Philosophical Transactions of the Royal Society A},
  volume={380},
  number={2227},
  pages={20210246},
  year={2022},
  publisher={The Royal Society}
}

@article{mediano2021towards,
  title={Towards an extended taxonomy of information dynamics via integrated information decomposition},
  author={Mediano, Pedro AM and Rosas, Fernando E and Luppi, Andrea I and Carhart-Harris, Robin L and Bor, Daniel and Seth, Anil K and Barrett, Adam B},
  journal={arXiv preprint arXiv:2109.13186},
  year={2021}
}

@article{gutknecht2025shannon,
  title={Shannon invariants: {A} scalable approach to information decomposition},
  author={Gutknecht, Aaron J and Rosas, Fernando E and Ehrlich, David A and Makkeh, Abdullah and Mediano, Pedro AM and Wibral, Michael},
  journal={arXiv preprint arXiv:2504.15779},
  year={2025}
}

@article{venkatesh2023gaussian,
  title={Gaussian partial information decomposition: {Bias} correction and application to high-dimensional data},
  author={Venkatesh, Praveen and Bennett, Corbett and Gale, Sam and Ramirez, Tamina and Heller, Greggory and Durand, Severine and Olsen, Shawn and Mihalas, Stefan},
  journal={Advances in Neural Information Processing Systems},
  volume={36},
  pages={74602--74635},
  year={2023}
}

@article{Liardi2025SimplePS,
  title={Simple physical systems as a reference for multivariate information dynamics},
  author={Alberto Liardi and Madalina I. Sas and George Blackburne and William J. Knottenbelt and Pedro A. M. Mediano and Henrik J. Jensen},
  journal={Chaos},
  year={2025},
  volume={35 10},
  url={https://api.semanticscholar.org/CorpusID:277780785}
}

\clearpage

\appendix

\section*{Appendix: Proofs}

We first present the coarse-grained PID redundancy lattice introduced by ~\cite{rosas2020reconciling} which, instead of decomposing the mutual information down to each individual source variable, is limited to collections of variables larger than a given size $k$, by using the notion of $k$\textsuperscript{th}-order synergy:
\begin{equation}\label{eqn:synk}
  \textnormal{\texttt{Syn}}^{(k)}(\bm X ; Y)\coloneqq \sum_{\bm\alpha\in\mathcal{S}^{(k)}} I_\partial^{\bm\alpha}(\bm X ; Y)
\end{equation}
where  $I_\partial^{\bm\alpha}(\bm X ; Y)$ represents the information that the set (or set of sets) of variables $\bm\alpha$ provide redundantly, but which is not provided by any smaller subsets. For example, $I_\partial^{\bm\alpha}(\bm X ; Y)$ corresponds to  $\{\{12\}\}$ if $\alpha=\{1,2\}$.

Intuitively, $\texttt{Syn}^{(k)}(\bm X; Y)$ corresponds to the information about the target that is provided by the
whole $\bm X$ but is not contained in any set of $k$ or less parts when considered separately. 
For example, for $n=2$ and $k=1$, we obtain the standard synergy $\mathcal{S}^{(1)} = \{ \{12\} \}$, 
and for $n=3$ and $k=1$ we have $\splitatcommas{\mathcal{S}^{(1)} = \{ \{12\}, \{13\}, \{23\}, \{12\}\{13\}, \{12\}\{23\}, \{13\}\{23\}, \{12\}\{13\}\{23\}, \{123\}\}}$.

The goal of this paper was to define a family of measures $\Psi^{(k,q)}$, where $q
= 0, \dots, n-1$ is the approximation order, such that $\Psi^{(k,0)} =
\Psi^{(k)}$ and $\Psi^{(1,n-1)} = \texttt{Syn}^{(1)}$.\footnote{The lattice expansion can be applied with $k > 1$, although in this case $\Psi^{(k,n-1)} \neq \texttt{Syn}^{(k)}$.}

To make the following proofs simpler, it's useful to visualise the involved
information-theoretic quantities as decomposed in a ``PID basis'' -- i.e. as
coefficients multiplying some ``basis vectors'' that are the PID atoms. In
general, any decomposable information-theoretic functional $F$ can be written as
\begin{align}
    F(\bm X; Y) = \sum_{\bm\alpha\in\mathcal{A}} a_{\bm\alpha} \PI{\bm\alpha}(\bm X; Y) ~ ,
\end{align}
with a given set of coefficients $a_{\bm\alpha}$. For example, in the case of $\Psi^{(1)}$ for $n = 3$, we have
\begin{align*}
  \def\hsep{1.5cm}
  a_{\bm\alpha} =
  \begin{cases} 
    \makebox[\hsep][l]{\hphantom{-}1} \text{if} \; \bm\alpha\in\mathcal{S}^{(1)} \\
    \makebox[\hsep][l]{-1} \text{if} \; \bm\alpha = \{i\}\{j\}, \; \forall i,j \\
    \makebox[\hsep][l]{-2} \text{if} \; \bm\alpha = \{1\}\{2\}\{3\} \\
    \makebox[\hsep][l]{\hphantom{-}0} \text{otherwise}
  \end{cases}
\end{align*}

Next, let us define $\mathcal{M}^n$ to be the set of antichains (i.e. subsets of the lattice made of elements which cannot be compared) that contain only singletons (i.e. sets with a single element). For example,
\begin{align*}
  \mathcal{M}^2\!&=\!\{ \{1\}, \{2\}, \{1\}\{2\} \} ~ ,\\
  \mathcal{M}^3\!&=\!\{ \{1\}, \{2\}, \{3\}, \{1\}\{2\}, \{1\}\{3\}, \{2\}\{3\}, \{1\}\{2\}\{3\} \} ~ .
\end{align*}
Within $\mathcal{M}^n$, the partial ordering of the redundancy lattice reduces to
\begin{align}
  \bm\alpha \preceq \bm\beta \Longleftrightarrow \bm\beta \subseteq \bm\alpha
\end{align}

The pair $(\mathcal{M}^n, \preceq)$ is a semilattice (where the greatest lower
bound always exists, but the least upper bound doesn't), which is a subset of
$\mathcal{A}^n$.

The following lemma will be instrumental in formulating the lattice expansion:

\begin{lemma}\label{lemma:singletons}

Consider two antichains $\bm\alpha, \bm\beta \in \mathcal{M}^n$ with
$|\bm\alpha| = |\bm\beta| < n$, and let $\bm\gamma = \bm\alpha \meet \bm\beta$,
and $F = \IR{\bm\alpha} + \IR{\bm\beta}$. Then, the following properties hold:

\begin{enumerate}

  \item $\bm\gamma \in \mathcal{M}^n$, and $\bm\gamma = \bm\alpha \bigcup
  \bm\beta$.

  \item The PID coefficients $a_{\bm\sigma}$ for $F$ satisfy
  \begin{align*}
    \def\hsep{1.5cm}
    a_{\bm\sigma} = 
    \begin{cases} 
      \makebox[\hsep][l]{1} \text{if} \; \bm\gamma \prec \bm\sigma \preceq \bm\alpha,\bm\beta \\
      \makebox[\hsep][l]{2} \text{if} \; \bm\sigma \preceq \bm\gamma \\
      \makebox[\hsep][l]{0} \text{otherwise}
    \end{cases}
  \end{align*}
\end{enumerate}
\end{lemma}
\begin{proof}
For convenience, recall the definition of the partial order in the redundancy
lattice~\cite{williams2010nonnegative}:
\begin{align*}
  \bm\alpha \preceq \bm\beta \Longleftrightarrow \forall b \in \bm\beta, \exists a\in\bm\alpha, a \subseteq b ~ .
\end{align*}
For the first property, we begin by proving that $\bm\alpha\bigcup\bm\beta
\subseteq \bm\gamma$. For this, note that since all $a \in \bm\alpha$ has $|a|
= 1$, the $\subseteq$ in the expression above has to be an equality (since empty
sets are not allowed~\cite{williams2010nonnegative}). Therefore, since
$\bm\gamma \preceq \bm\alpha,\bm\beta$ we have that $\forall
a\in\bm\alpha,\bm\beta, a\in\bm\gamma$; and thus $\bm\alpha\bigcup\bm\beta
\subseteq \bm\gamma$. The equality can be proven using the fact that
$\bm\gamma$ must be the greatest lower bound of $\bm\alpha$ and $\bm\beta$: if
it was the case that $\bm\alpha\bigcup\bm\beta \subset \bm\gamma$, then there
would exist $\bm\sigma = \bm\alpha\bigcup\bm\beta$ such that $\bm\gamma \prec
\bm\sigma \preceq \bm\alpha,\bm\beta$, resulting in a contradiction.

The second property follows directly from the definition of $\bm\gamma$ and
Eq.~\eqref{eq:redpartial}.

\end{proof}

This guarantees that if you take a sum
like $\sum_{{\bm\alpha\in\mathcal{M}^n,\\|\bm\alpha| = r}}$, then all
atoms down to the next level ($\bm\beta\in\mathcal{M}^n, |\bm\beta| = r+1$)
have the same coefficient.

In turn, this means that if we want to guarantee
that $a_{\bm\alpha} = c$ for all $\bm\alpha\in\mathcal{A}^n$, it suffices to
guarantee that $a_{\bm\beta} = c$ for all $\bm\beta\in\mathcal{M}^n$. 
Now we use Lemma~\ref{lemma:singletons} to prove a useful property of
$\Psi^{(1)}$.

\begin{lemma}\label{lemma:psi}
The PID coefficients $a_{\bm\alpha}$ for $\Psi^{(1)}$ satisfy $a_{\bm\alpha} =
1-g(\bm\alpha)$, where $g(\bm\alpha)$ is the number of singleton elements in
$\bm\alpha$.
\end{lemma}
\begin{proof}
The $+1$ term comes from the $I(V_t; V_{t'})$ in $\Psi$, which includes all
nodes in the lattice. In addition, any given $\bm\alpha$ may be affected by the
negative $I(X_t^i; V_{t'}) = \IR{\{i\}}(\bX; V_{t'})$ in $\Psi$. Following a
logic similar to Lemma~\ref{lemma:singletons}, $\bm\alpha \preceq \{i\}$ iff
$\{i\} \in \bm\alpha$. Since the sum in $\Psi^{(1)}$ runs for all individual
sources, any particular atom $\bm\alpha$ will be subtracted once for each
$\{i\}$ it precedes -- or, equivalently, the number of singleton elements in
$\bm\alpha$.
\end{proof}

The coefficients allow us to formulate the lattice expansion for the $\Psi$ emergence criterion introduced in Eq.~\ref{eq:expansion}, which we repeat here for convenience:
\begin{align*}
  \Psi^{(k,q)} = \Psi^{(k)} + \sum_{r=n-q+1}^n \sum_{\substack{\bm\alpha\in\mathcal{M}^n\\|\bm\alpha| = r}} C^n_{q,r} \IR{\bm\alpha}(\bm X_t; V_{t'}) ~ ,
\end{align*}
with the coefficients $C^n_{q,r}$ defined recursively as
\begin{align*}
  C^n_{q,r} &= r-1 - \sum_{s=n-q+1}^{r-1} C^n_{q,s} \binom{r}{s} \\
  C^n_{q,n-q+1} &= n-q,
\end{align*}
where $n$ is the system size, $q \in [1, n-1]$ the order of the expansion,
$r \in [n-q+1, n]$ the cardinality of the set of sources being considered, 
and $\binom{r}{s}$ the binomial operator.

\begin{lemma}\label{lemma:expansion}
The quantity $\Psi^{(k,q)}$ in Eq.~\eqref{eq:expansion} is a valid criterion
for emergence, it is greater than $\Psi^{(k)}$, and for the $k = 1$ case the
full expansion for $q = n-1$ satisfies $\Psi^{(1,n-1)} =
\textnormal{\texttt{Syn}}^{(1)}$.
\end{lemma}
\begin{proof}
For a given $n,q$, we begin by correcting the highest level in the lattice
(i.e. the level with the lowest $r$). As per Lemma~\ref{lemma:psi}, nodes with
$r$ singletons need to be corrected by a factor $r-1$; and thus the coefficient
for the lowest $r$ in the expansion satisfies $C^n_{q,n-q+1} = n-q$.

Crucially, Lemma~\ref{lemma:singletons} guarantees that once the atoms in
$\mathcal{A}^n$ with exactly $r$ singletons have been corrected, then all atoms
between them and the atoms with $r+1$ singletons are corrected too. This allows
us to only consider atoms in $\mathcal{M}^n$ with different cardinalities, and
guarantees that all atoms in the full PID lattice will be covered.

After atoms with cardinality $c$ have been corrected, we can consider the set
of atoms in $\mathcal{M}^n$ with cardinality $c+1$. These atoms will have a
coefficient $1-r$ (as per Lemma~\ref{lemma:psi}), plus the contributions from
the terms in the expansion with lower $r$. In general, an atom with cardinality
$r$ will precede $\binom{r}{s}$ atoms with cardinality $s < r$, which need to
be accounted for in the expansion -- and thus the recursive expression for
$C^n_{q,r}$.

Finally, since, for all $q$, $\Psi^{(k,q)}$ has no positive $a_{\bm\alpha}$ for
$\bm\alpha \notin \mathcal{S}^{(1)}$, the same arguments for $\Psi^{(k)}$ being
a valid criterion for emergence apply to $\Psi^{(k,q)}$.
\end{proof}

\end{document}